\RequirePackage{afterpackage}
\AfterPackage{amsthm}{
\RequirePackage{hyperref}
\RequirePackage[nameinlink,capitalize]{cleveref}
}

\pdfoutput=1
\documentclass[a4paper,USenglish]{lipics-v2019}

\nolinenumbers
\hideLIPIcs

\usepackage{bm, amsmath,  amssymb, amsfonts}

\usepackage{comment}
\usepackage{framed}
\usepackage[framemet hod=tikz]{mdframed}
\usepackage{appendix}
\usepackage{graphicx}
\usepackage{color}
\usepackage{varwidth}
\usepackage{wrapfig}
\usepackage[T1]{fontenc}

\usepackage[textsize=tiny]{todonotes}

\usepackage{xspace}
\usepackage{xifthen}

\usepackage{algorithm}
\usepackage[noend]{algpseudocode}

\newtheorem{property}[theorem]{Property}
\newtheorem{observation}[theorem]{Observation}
\usepackage{mathtools}

\newcommand{\cO}{O}

\renewcommand{\inf}{\infty}
\newcommand{\ignore}[1]{}

\algnewcommand\algorithmicswitch{\textbf{switch}}
\algnewcommand\algorithmiccase{\textbf{case}}

\algdef{SE}[SWITCH]{Switch}{EndSwitch}[1]{\algorithmicswitch\ #1\ \algorithmicdo}{\algorithmicend\ \algorithmicswitch}%
\algdef{SE}[CASE]{Case}{EndCase}[1]{\algorithmiccase\ #1}{\algorithmicend\ \algorithmiccase}%
\algtext*{EndSwitch}%
\algtext*{EndCase}%

\newcommand{\LOCAL}{LOCAL\xspace}

\newcommand{\eps}{\varepsilon}

\DeclareMathOperator{\polylog}{polylog}
\DeclareMathOperator{\poly}{poly}

\newcommand{\complexityclass}[2][]{\ensuremath{\mathsf{#2}\ifthenelse{\isempty{#1}}{}{(#1)}}}

\newcommand{\hide}[1]{}

\DeclareMathOperator{\dis}{dis}

\renewcommand{\tilde}{\widetilde}

\title{Improved MPC Algorithms for MIS, Matching, and Coloring  on Trees and Beyond}

\author{Mohsen Ghaffari}{ETH Zurich, Switzerland}{ghaffari@inf.ethz.ch}{}{}
\author{Christoph Grunau}{ETH Zurich, Switzerland}{cgrunau@student.ethz.ch}{}{}
\author{Ce Jin}{Tsinghua University, China}{cejin@mit.edu}{}{}

\funding{This work was supported in part by funding from the European Research Council (ERC) under the European Union’s Horizon 2020 research and innovation programme (grant agreement No. 853109), and the Swiss National Foundation (project grant 200021-184735).}

\authorrunning{M.\,Ghaffari,  C.\,Grunau, and  C.\,Jin}

\Copyright{Mohsen Ghaffari, Christoph Grunau, and Ce Jin} 


\EventEditors{Hagit Attiya}
\EventNoEds{1}
\EventLongTitle{34rd International Symposium on Distributed Computing (DISC 2020)}
\EventShortTitle{DISC 2020}
\EventAcronym{DISC}
\EventYear{2020}
\EventDate{October 12--18, 2020}
\EventLocation{Virtual Conference}
\EventLogo{}
\SeriesVolume{179}
\ArticleNo{34}
\begin{document}
\maketitle

\begin{abstract}

We present $O(\log\log n)$ round scalable Massively Parallel Computation algorithms for maximal independent set and maximal matching, in trees and more generally graphs of bounded arboricity, as well as for coloring trees with a constant number of colors. Following the standards, by a \emph{scalable} MPC algorithm, we mean that these algorithms can work on machines that have capacity/memory as small as $n^{\delta}$ for any positive constant $\delta<1$. Our results improve over the $O(\log^2\log n)$ round algorithms of Behnezhad et al.\ [PODC'19]. Moreover, our matching algorithm is presumably optimal as its bound matches an $\Omega(\log\log n)$ conditional lower bound of Ghaffari, Kuhn, and Uitto [FOCS'19].
\end{abstract}

\section{Introduction and Related Work}
We present improved algorithms for some of the central graph problems in the study of large-scale algorithms---namely maximal matching and matching approximation, maximal independent set, and graph coloring---in the Massively Parallel Computation (MPC) setting. We first review the related model and known results, and then we present our contributions. 

\subsection{Massively Parallel Computation}

The MPC model was introduced by Karloff et al.~\cite{karloff2010model}, and after some refinements~\cite{goodrich2011sorting, beame2014skew}, it has by now become the standard theoretical abstraction for studying algorithmic aspects of large-scale data processing. This model captures the commonalities of popular practical frameworks such as MapReduce~\cite{dean2008mapreduce}, Hadoop~\cite{white2012hadoop}, Spark~\cite{zaharia2010spark} and Dryad~\cite{isard2007dryad}. 

In the particular case of MPC for graph problems, we assume that our $n$-node $m$-edge input graph is partitioned arbitrarily among a number of machines, each with memory $S$. This \emph{local memory} $S$ is assumed to be considerably smaller than the entire graph size. Over all the machines, the summation of memories (called the \emph{global memory}) should be at least $\Omega(m+n)$, and often assumed to be at most $\tilde{O}(m+n)$. 
The machines communicate in synchronous message-passing rounds, subject to only one constraint: per round, each machine can only send or receive a data of size at most $S$, which is all that it could fit in its memory. The objective is to solve various graph problems in the fewest possible number of rounds.

As the graph sizes are getting larger and larger (at a pace that exceeds that of single computers), we would like to have algorithms where each machine's local memory is much smaller than the entire graph size. However, dealing with smaller local memory is known to increase the difficulty of the problem, as we soon review in state-of-the-art. 
Based on this, the state-of-the-art in MPC algorithms can be divided into three regimes, depending on how the local memory $S$ compares with the number of vertices $n$: (A) The strongly super-linear memory regime where $S\geq n^{1+\delta}$ for a constant $\delta>0$, (B) The near-linear memory regime where $S= n \poly(\log n)$, and (C) the strongly sublinear memory regime where $S\leq n^{\delta}$ for a positive constant $\delta <1$. We note that algorithms in the last category are the most desirable ones and they are usually referred to as \emph{scalable MPC algorithms}. 

\subsection{State-of-the-Art}
For the problems under consideration in this paper, simple $O(\log n)$ time algorithms follow from classic literature in the LOCAL model and PRAM \cite{luby1986simple, alon1986fast}. The primary objective in MPC is to obtain significantly faster algorithms than the PRAM counterpart, ideally just constant rounds or $O(\log\log n)$ rounds. Over the past few years, there has been progress toward this, gradually moving toward the lower memory regimes.

For the strongly superlinear memory regime, Lattanzi et al.~\cite{LattaMSV2011Filtering} presented $O(1)$-round algorithms for maximal matching and maximal independent set. Progress to lower memory regimes was slow afterward. But that changed with a breakthrough of Czumaj et al.~\cite{czumaj2019roundsiam}, who gave an $O(\log^2\log n)$ round algorithm for $(2+\eps)$-approximation of maximum matching in the near-linear memory regime. Shortly after, for the same near-linear memory regime, Assadi et al.~\cite{Assadi2017CoresetsME} gave an $O(\log\log n)$ round algorithm for $(1+\eps)$-approximation of maximum matching and constant approximation of minimum vertex cover and independently Ghaffari et al.~\cite{GhaffGKMR2018Improved} gave $O(\log\log n)$ round algorithms for $(1+\eps)$-approximation of maximum matching, $(2+\eps)$-approximation of minimum vertex cover, and maximal independent set. Finally, Behnezhad et al.~\cite{BehneHH2019Exponentially} gave an $O(\log\log n)$ round algorithm for maximal matching. 

For the much more stringent strongly sublinear memory regime, there has also been some progress but much slower: For the case of general graphs, Ghaffari and Uitto~\cite{GhaffU2019Sparsifying} provide an $\tilde{O}(\sqrt{\log \Delta}+\log \log n)$ algorithm for graphs with maximum degree $\Delta$, which remains the best known.
When $\Delta > n^{\Omega(1)}$, the only known algorithms that outperform this bound and reach the $\poly(\log\log n)$ regime are results for special graph families, namely trees~\cite{brandt2019breaking,hajiaghayi1,hajiaghayi2} and more generally graphs of small, e.g., polylogarithmic, arboricity~\cite{BrandFU2018Matching, BehneDHK2018Massively,BehneBDFHKU2019Massively}. Recall that the \emph{arboricity} $\alpha$ of a graph is the minimum number of forests into which its edges can be partitioned. Many natural graph classes, such as planar graphs and graphs with bounded treewidth, have small arboricity.    Concretely, Behnezhad~et~al.~\cite{BehneBDFHKU2019Massively} provide Maximal Matching and MIS algorithms that run in $O((\log\log n)^2)$ rounds and leave remaining graphs with maximum degree $\poly(\max(\log n, \alpha))$. By invoking the aforementioned algorithm of Ghaffari and Uitto~\cite{GhaffU2019Sparsifying}, they obtain maximal matching and MIS algorithms that have round complexity $O(\sqrt{\log \alpha} \log \log \alpha + (\log\log n)^2)$. Finally, a result of Ghaffari, Kuhn, and Uitto~\cite{GhaffKU2019Conditional} shows a conditional lower bound (for \emph{component-stable} MPC algorithms\footnote{Most of the MPC algorithms in the literature are component-stable. One exception is a recent result due to Czumaj, Davies, and Parter \cite{czumaj2019graph}, showing a deterministic MPC algorithm for MIS and maximal matching using strongly sublinear local memory.}) of $\Omega(\log\log n)$ for maximal matching and MIS in this strongly sublinear memory regime. In essence, they show that component-stable MPC algorithms in the strongly sublinear memory regime cannot be more than exponentially faster compared to their LOCAL-model counterpart, and then they use lower bounds in the LOCAL model. Their result is conditioned on the widely believed $\Omega(\log n)$ complexity of the connectivity problem in the strongly sublinear memory regime. We note that although an $\Omega(\log n)$ round complexity lower bound is widely believed to be true for the connectivity problem in the strongly sublinear memory regime, even for distinguishing an $n$-node cycle from two $n/2$-node cycles, proving any super-constant lower bound would imply $\mathsf{P}\not \subset \mathsf{NC}^1$, a major breakthrough in circuit complexity~\cite{roughgarden2016shuffles}. Concretely, in the case of maximal matching and MIS, the results of Ghaffari et al.~\cite{GhaffKU2019Conditional} invoke the $\Omega(\sqrt{\log n/\log \log n})$ round LOCAL-model lower bound of Kuhn, Moscibroda, and Wattenhofer~\cite{Kuhn:2016:LCL:2906142.2742012} to show that an $o(\log\log n)$-round MPC algorithm for maximal matching or MIS would imply an $o(\log n)$ round algorithm for connectivity, which would break the conjectured $\Omega(\log n)$ complexity. We emphasize that in the case of maximal matching, this $\Omega(\log\log n)$ round complexity lower bound for MPC algorithms holds even if the input graph is a tree.

\subsection{Our Contribution}
Our main contribution is an improvement over the algorithm of Behnezhad~et~al.\ \cite{BehneBDFHKU2019Massively}. We obtain MPC algorithms in the strongly sublinear memory regime that solve maximal matching and MIS in $O(\log\log n)$ rounds, in trees and constant-arboricity graphs. And more generally as a function of the arboricity $\alpha$, our maximal matching and MIS algorithms run in $O(\sqrt{\log \alpha} \log \log \alpha + \log\log n)$ rounds, where the first term again comes from invoking the algorithm of Ghaffari and Uitto~\cite{GhaffU2019Sparsifying} for remaining graphs with maximum degree $\poly(\alpha)$.

\begin{theorem}
\label{thm:lowarb:main}
	There is an $O(\sqrt{\log \alpha} \log \log \alpha + \log\log n)$ round MPC algorithm using $n^{\delta}$ local memory and $\tilde O(n+m)$ global memory 	 that with high probability computes a maximal matching and a maximal independent set of a graph with arboricity $\alpha$.
\end{theorem}
    We note that on graphs with polylogarithmic arboricity, our bound provides a quadratic improvement over the previous algorithms \cite{BehneDHK2018Massively,BrandFU2018Matching,BehneBDFHKU2019Massively}. 	Moreover, our algorithm for maximal matching is presumably optimal, as the $\Omega(\log \log n)$ conditional lower bound shown by Ghaffari, Kuhn, and Uitto \cite{GhaffKU2019Conditional} holds even on trees (where $\alpha=1$).
	
Our second contribution is providing a similarly fast algorithm for $4$-coloring trees. 	
\begin{theorem}
\label{thm:4-coloring}
There is an $O(\log\log n)$ round MPC algorithm using $n^{\delta}$ local memory and $O(n+m)$ global memory
	 that with high probability computes a $4$-coloring of any tree.
\end{theorem}
This algorithm also matches a conditional lower bound of $\Omega(\log\log n)$, which follows from the conditional impossibility results of \cite{GhaffKU2019Conditional}, showing that strongly sublinear memory MPC algorithms cannot be more than exponentially fast compared to their LOCAL model counterpart, and the $\Omega(\log n)$ LOCAL-model lower bound for $O(1)$-coloring trees of Linial~\cite{linial}.

	\subsection{Preliminaries and Notations}
	Let $[n]$ denote the set $\{1,2,\dots,n\}$. Let $G=(V,E)$ be an undirected graph. 
	Let $\deg_G(u)$ denote the degree of a vertex $u$ in graph $G$, and let $\dis_G(u,v)$ denote the minimum number of edges in a path connecting $u$ and $v$ (when $G$ is clear from the context, we may also write $\deg(u)$ and  $\dis(u,v)$ for simplicity).  For a vertex subset $V' \subseteq V$, define $\dis_G(u,V')=\min_{v\in V'}\dis_G(u,v)$. 
	The \emph{$k$-hop neighborhood} of a vertex $v$ in $G$ is the subgraph containing all vertices $u$ satisfying $\dis_G(u,v)\le k$, together with every edge that lies on some path starting from $v$ with length at most $k$.
	The $k$-hop neighborhood of a vertex set $V'\subseteq V$ is the union of $k$-hop neighborhoods of $v\in V'$. We use $G[V']$ to denote the subgraph of $G$ induced by the vertex subset $V' \subseteq V$. For simplicity we use $\dis_{V'}(u,v)$ as a shorthand for $\dis_{G[V']}(u,v)$, when the underlying graph $G$ is clear from the context; similarly, we can define  $\deg_{V'}(u)$.
	
	The \emph{arboricity} $\alpha$ of a graph is the minimum number of forests into which its edges
can be partitioned.  Nash-Williams~\cite{Nash-1964Decomposition} showed that the arboricity is equal to the maximum value of $\lceil m_S/(n_S-1) \rceil$, where $m_S,n_S$ are the number of edges and vertices in any subgraph $S$.

\section{MIS and Matching}
In this section we prove \cref{thm:lowarb:main}. In \cref{sec:lowarb:overview}, we review the previous work by Behnezhad et al. and informally describe our new techniques for improving their algorithm. 
In \cref{sec:lowarb:defini}, we give formal definitions and describe the structure of our pipelined MPC algorithm.
Then in \cref{sec:2.3} and \cref{sec:lowarb:details-algo}, we present the implementation and loop invariants of our algorithm in detail.
In \cref{sec:lowarb:memory}, we analyze the space complexity of the MPC algorithm.

\subsection{Overview}
\label{sec:lowarb:overview}
\subparagraph*{Review of Behnezhad~et~al.'s algorithm}We first briefly review the main ideas of the previous MPC algorithm for Maximal Independent Set (MIS) and Maximal Matching (MM) by Behnezhad~et~al.\ \cite{BehneBDFHKU2019Massively, BehneDHK2018Massively}. Their result is based on the \LOCAL algorithm due to Barenboim~et~al.\ \cite{barenboim2016locality}. This LOCAL algorithm is formally stated in \cref{lem:local}. 
In a graph with small arboricity $\alpha$ and large maximum degree $\Delta \ge \poly(\alpha, \log n)$, the \LOCAL algorithm finds a matching (or an independent set) in $O(1)$ rounds and removes the involved vertices from the graph, so that the number of high-degree vertices (i.e., with degree $>\sqrt{\Delta}$) decreases by a factor of $\Delta$ in the remaining graph.

By repeating this \LOCAL algorithm $O(\log_{\Delta}n)$ times, all high-degree vertices are eliminated, that is, they either get removed from the graph or become low-degree. Thus, the remaining graph has maximum degree at most $\sqrt{\Delta}$. We call this procedure a \emph{phase}. 
After $O(\log \log \Delta)$ phases, the maximum degree drops below $\poly(\alpha,\log n)$, and we switch to the algorithm by Ghaffari and Uitto \cite{GhaffU2019Sparsifying} to find a maximal matching (or MIS) of the remaining low-degree graph.
Combining with the partial solution obtained in previous phases, we then obtain a maximal matching (or MIS) of the input graph.

To implement the algorithm efficiently in the MPC model, Behnezhad~et~al.\ use the, by-now standard, graph exponentiation technique~\cite{Lenzen2010brief, Ghaffari2017MISclique}:
For each vertex $v$, we store its $d$-hop neighborhood in one machine, with $d$ initially being $1$. In $O(1)$ MPC rounds we can double the radius to $2d$, by requesting and collecting the $d$-hop neighborhoods of all vertices $u$ in the $d$-hop neighborhood of $v$.
Hence in $O(\log d)$ MPC rounds we can collect the
 $d$-hop neighborhood of each vertex, which has at most $\Delta^d$ edges and can fit into the $n^{\delta}$ local memory when $d=\delta \log_{\Delta} n$. 
 Then we can simulate $O(\log_{\Delta} n)$ \LOCAL rounds in this phase using only $O(1)$ MPC rounds.
As each phase requires $O(\log d) \le O(\log \log n)$ MPC rounds for graph exponentiation, the total round-complexity of this MPC algorithm is $O(\log \log n \log \log \Delta) \le  O(\log^2 \log n)$.

The algorithm described above would need $\tilde O(n^{1+\delta}+m)$ global memory. In order to reduce it to $\tilde O(n+m)$, Behnezhad~et~al.\ used an additional idea: 
As our objective in a phase is to make high-degree vertices disappear, we can without loss of generality assume that the \LOCAL algorithm only removes vertices that are in the $O(1)$-hop neighborhood of the high-degree vertices, and other vertices are considered as irrelevant and do not participate in the LOCAL algorithm. 
As the number of high-degree vertices decreases after every execution of the \LOCAL algorithm, the number of relevant vertices also decreases quickly. Then, we have enough space for expanding the stored neighborhood of each relevant vertex.

\subparagraph*{New techniques}
Our new idea is a \emph{pipelining} technique:
as more vertices become irrelevant in the current phase, we can start running the next phase on these vertices, \emph{concurrently} with the current phase which has not necessarily finished completely.
In our algorithm, after starting phase $(\ell-1)$, we wait for $O(1)$ MPC rounds and then start phase $\ell$. 
There are $O(\log \log \Delta)$ phases in total, each taking $O(\log \log n)$ MPC rounds, so the round complexity is $O(\log \log \Delta)\cdot O(1) + O(\log \log n) = O(\log \log n)$ in total.
Our pipelined algorithm produces exactly the same result as the unpipelined version does (using the same random bits).

When running phase $\ell$, we need to deal with \emph{pending} vertices that have not finished their computation in previous phases $1,2,\dots,\ell-1$.
As it is still unknown whether they will be removed by the start of phase $\ell$, the messages sent from these vertices during phase $\ell$ are temporarily marked as \emph{pending}.
If a vertex receives a pending message, then its state also becomes pending, and so on. 
When simulating the LOCAL computation in phase $\ell$, we are only able to simulate the behaviour of non-pending vertices; only upon finishing the simulation of previous phases $1,2,\dots,\ell-1$ for some vertex, the initial state of this vertex in phase $\ell$ becomes known, which allows us to resume the phase $\ell$ computation around its neighborhood and remove the pending marks. 

Now we take a closer look at what happens when we perform graph exponentiation.
In the ideal scenario, after phase $\ell-1$ completely finishes, the maximum degree drops below $\Delta_{\ell-1}:=\Delta^{1/2^{\ell-1}}$, allowing us to store the $(\delta \log_{\Delta_{\ell-1}} n)$-hop neighborhood of a vertex into the local memory.
However, in the actual situation, when expanding the neighborhood we may encounter pending vertices, which may not have degree upper bounded by $\Delta_{\ell-1}$.
It would be problematic if we simply excluded these vertices when we expand the neighborhood, as they may later have degree reduced to the interval $(\Delta_{\ell}, \Delta_{\ell-1}]$ and become relevant in phase $\ell$. 
So we have to store them as well, and this poses challenges to bounding the local memory and the global memory of our algorithm---this is the most technical part of our analysis.

In this overview, we provide some intuition on how we analyze the global memory used for storing the neighborhoods. 
We will analyze the neighborhoods of non-pending vertices and of pending vertices separately.
In phase $\ell$, after repeating the \LOCAL algorithm $k$ times,  the number of non-pending vertices $v$ with degrees in $(\Delta_{\ell},\Delta_{\ell-1}]$ can be bounded by $n/\Delta_{\ell-1}^k$.
The $d$-hop neighborhood (we will maintain $d=ck$ for some constant $c\in (0,1)$) of such a non-pending vertex $v$  only contains vertices $u$ with degree at most $\Delta_{\ell-1}$, since otherwise $u$ would still be waiting for phase $\ell-1$ to finish, causing $v$ to become a pending vertex.
Hence the total size of $d$-hop neighborhoods of these non-pending vertices $v$ is roughly at most $(n/\Delta_{\ell-1}^k)\cdot \Delta_{\ell-1}^d<n$.
Now we consider the pending vertices. For any pending vertex $v$, it must be within $k$ distance from a vertex $u$ which is currently stuck at phase $\ell-q$, for some $q\ge 1$.
We pick such $u$ with the biggest $q$, so $v$ can be reached from $u$ by a $k$-step path on which all vertices have degree $\le \Delta_{\ell-q}$ (since if any of them had degree greater than $\Delta_{\ell-q}$, it would be stuck at a phase earlier than $\ell-q$, contradicting the maximality of $q$).
Hence we can bound the number of pending vertices $v$ by $\sum_{q\ge 1} H_{\ell-q} \cdot \Delta_{\ell-q}^{k}$, where $H_{\ell-q}$ denotes the number of vertices currently stuck at phase $\ell - q$, and the total size of neighborhoods of these $v$ can be bounded by a similar summation. In order to show a near-linear upper bound, we want to make this summation dominated by a geometric series.  To do this, we use a similar argument to show an upper bound on $H_{\ell-q}$, and hence after setting appropriate parameters, we can use \emph{induction on $\ell$} to establish the desired memory bound.

	\subsection{Definitions and algorithm structure}
	\label{sec:lowarb:defini}
	
	Suppose each machine in the MPC model has $n^\delta$ local memory, for some constant  $0<\delta<1$. 
	We will use the following \LOCAL algorithm \cite{barenboim2016locality} as a black box.
	
	\begin{lemma}[{\cite[Theorem~7.2]{barenboim2016locality}}, restated\footnote{{The proof of {\cite[Theorem~7.2]{barenboim2016locality}} presented an $O(1)$-round algorithm that reduces the number of vertices with degree greater than $t\alpha$ by a $t^{1/7}$ factor, for any parameter $t\ge \poly(\alpha,\log n)$. By choosing $t = \sqrt{\Delta}/\alpha$ and repeating the algorithm $O(1)$ times we obtain the claimed statement.}}; see also {\cite[Section~5.2]{BehneDHK2018Massively}}]
	\label{lem:local}
	Let $\alpha$ and $\Delta$ be parameters satisfying $\Delta \ge \poly(\alpha, \log n)$.
	Suppose the input graph $G$ has arboricity at most $\alpha$ and maximum degree at most $\Delta$.
Let $H_G:=\{v: \deg_G(v)> \sqrt{\Delta}\}$ denote the set of high-degree vertices in graph $G$.
	
	There is a \LOCAL algorithm of round complexity $r=O(1)$ that computes a matching $M$ (or an independent set $I$) of $G$, such that  $|H_{G'}|\le |H_G|/\Delta$ w.h.p., where $G'$ is the induced subgraph of $G$ where matched vertices in $M$ (or vertices in $I$ and their neighbors) are removed. 
	Moreover, the algorithm has the following properties.
	\begin{enumerate}[(1)]
	    \item  The state description and random bits for each vertex, as well as communication along every edge in each round, have length at most $\polylog(n)$ bits. The space usage for computing the new state of vertex $v$ is $\deg_G(v)\cdot \polylog(n)$, and the space usage for computing the message sent along each edge is $\polylog(n)$.
	    \label{property1}
	    \item Let $H^+$ denote the $r$-hop neighborhood of $H_G$ in $G$.  Then vertices outside $H^+$ are not affected by the algorithm (i.e., they will stay in graph $G'$, and their degrees in $G'$ are unchanged).
	    \label{property2}
	\end{enumerate}
	\end{lemma}
Let  $\Delta_\ell := \Delta^{1/2^{\ell}}$, where $\Delta$ is the maximum degree of the input graph.
The unpipelined algorithm sequentially runs $L = O(\log \log \Delta)$ \emph{phases}, where phase $\ell$ ($\ell=1,2,\dots,L$) repeats the \LOCAL algorithm $O(\log_{\Delta_{\ell}} n)$ times with degree parameter $\Delta_{\ell-1}$, and reduces the maximum degree from $\Delta_{\ell-1}$ to $\Delta_{\ell}$.


Our pipelined MPC algorithm runs in $O(\log \log n)$ \emph{iterations} (starting from iteration 1), each taking $O(1)$ MPC rounds. 
In each iteration, there are multiple \emph{phases} concurrently being simulated by our MPC algorithm.
There is a small \emph{lag} $t=O(1)$ between the start of phase $\ell$ and $\ell+1$, i.e.,  the simulation of phase $\ell$ starts at iteration $j=(\ell-1) t +1$, by running $\textsc{Initialize}(\ell)$.  In each  iteration $j$, we perform $\textsc{Simulate}(\ell,j),\textsc{Update}$, and $\textsc{Expand}(\ell,j)$, concurrently for all active phases $\ell$ (i.e, phases that have already started).
 See \cref{algo:main} for the algorithm structure. 
 The value of parameter $t$ will be determined in \cref{sec:lowarb:memory}.
		\begin{algorithm}[H]
		\small
			\caption{Structure of our pipelined MPC algorithm}
		\label{algo:main}
		\begin{algorithmic}
		\State Number of phases $L = O(\log \log \Delta)$. 
		\State Lag parameter $t=O(1)$.
		\For {iteration $j\gets 1,2,\dots, O(\log \log n)$}
	        \State $L_{\mathsf{active}} \gets \min\{L, \lceil (j-1)/t\rceil \}$ \textit{(Number of phases that have already started)}
		    \If{$j = (\ell - 1)t+1$ for some $\ell \in [L]$}
		        \State $\textsc{Initialize}(\ell)$ \textit{(Phase $\ell$ starts in this iteration)}
	        \EndIf
            \State $\textsc{Simulate}(\ell,j)$  for all $\ell \in [L_{\mathsf{active}}]$ concurrently
            \State $\textsc{Update}$
            \State $\textsc{Expand}(\ell,j)$  for all $\ell \in [L_{\mathsf{active}}]$ concurrently
	    \EndFor
	    \State Switch to Ghaffari and Uitto's algorithm~\cite{GhaffU2019Sparsifying}
	    \end{algorithmic}
	\end{algorithm}
Our MPC algorithm maintains the subgraph induced by the remaining vertices. After iteration $j$ finishes, the remaining induced subgraph is denoted by $G(j)$. We have $G = G(0)\supseteq G(1) \supseteq \cdots$. 
In iteration $j$, concurrently for each active phase $\ell$,  $\textsc{Simulate}(\ell,j)$ simulates (part of) the \LOCAL computation in phase $\ell$, and removes the matched vertices (or vertices in the independent set and their neighbors) from the graph.
After removing the vertices, in $\textsc{Update}$ we compute the degrees of vertices in the remaining graph $G(j)$, and update some other information.
Then, in $\textsc{Expand}(\ell,j)$, we perform one graph exponentiation step and double the radius of the neighborhoods stored in memory.

\subsection{LOCAL simulation with incomplete information}
\label{sec:2.3}
As discussed in \cref{sec:lowarb:overview}, our pipelining algorithm is complicated by the presence of pending vertices. In this section we explain how to deal with them properly when simulating LOCAL algorithms.

When we simulate a LOCAL algorithm $A$ on an input graph $G$, we may not have the complete information of the initial states of vertices.
We say vertex $v$ is \emph{initially pending} if $v$'s initial state is unknown.\footnote{ Suppose LOCAL algorithm $A$ refers to the computation in phase $\ell$. Here, the \emph{state} of a vertex $v$ includes its degree, whether it has been removed, and some other information required by \cref{lem:local} Property (\ref{property1}).  If we have not finished the phase $(\ell-1)$ simulation for vertex $v$, then we do not know $v$'s state at the beginning of phase $\ell$, and we will consider $v$ to be in a special \emph{pending state}.}
We have the following simple observation:
\begin{observation}
\label{obs:pending}
Define a LOCAL algorithm $A'$, which is the same as $A$ except for the following additional rule: the messages sent from a pending vertex are marked as pending, and a vertex receiving a pending message becomes pending.

If $v$'s $R$-hop neighborhood in the input graph $G$ does not contain an initially pending vertex, then $v$'s state after running algorithm $A'$ for $R$ rounds is non-pending, and is the same as its state in algorithm $A$ after $R$ rounds.
\end{observation}



Now we introduce another simple lemma which will be helpful in reducing the space usage of our MPC simulation.
Consider a LOCAL algorithm $A$ in which vertices may get eliminated during execution: after a vertex is eliminated, it no longer participates in later rounds of the algorithm. 
It is well-known that we can simulate the behaviour of a vertex $v$ up to $R$ LOCAL rounds provided that we know $v$'s $R$-hop neighborhood in the input graph $G=(V,E)$.
However, if we somehow know a subset $U$ of vertices that will be eliminated by $A$, then we only need to collect $v$'s $R$-hop neighborhood in the induced subgraph $G[V\setminus U]$ (plus the messages sent from $U$ when they are alive), which could be much smaller than its $R$-hop neighborhood in $G$. 
\begin{lemma}
\label{lemma:died-simulation}
Suppose a LOCAL algorithm $A$ runs on an input graph $G=(V,E)$, and $U_0\subseteq V$ is the set of vertices that are eliminated during the execution of $A$. 

Let $U \subseteq U_0$, and $v\in V\setminus U$. Then, we can simulate the behaviour of $v$ in the first $R$ rounds of the algorithm $A$, provided that we have the information of:
\begin{enumerate}[(a)]
    \item The $R$-hop neighborhood of $v$ in the induced subgraph $G[V\setminus U]$.  \label{item:neighbor1}
    \item The complete communication history along every edge $(u,w)\in E$ where $u\in U$ and $w$ belongs to the neighborhood defined in (\ref{item:neighbor1}), until $u$ gets eliminated by $A$.
\end{enumerate}
\end{lemma}
The proof of this lemma is trivial.
\subsection{Algorithm in detail}
\label{sec:lowarb:details-algo}

Recall that each machine in the MPC model has $n^\delta$ memory. Define $s:=\lceil 10/\delta \rceil $.
Recall that $r=O(1)$ is the round complexity of the LOCAL algorithm from \Cref{lem:local}, and $G(j)$ is the graph induced by the remaining vertices at the end of iteration $j$.
Let $H_\ell(j)$ denote the set of vertices $v$ with $\deg_{G(j)}(v) > \Delta_{\ell}$, and let $H^+_\ell(j)$ be the $r$-hop neighborhood of $H_\ell(j)$ in graph $G(j)$.
From \cref{lem:local} (Property (\ref{property2})), we observe that vertices outside $H_{\ell}^+(j)$ are not affected by phase $1,2,\dots,\ell$ in the following iterations $j+1,j+2,\dots$.


  At the end of iteration $j$, we maintain the following invariants for all active phases $\ell$. 
\begin{property}
\label{prop:lowarb:invariant}
  Let $j=(\ell-1) t+i$ $(i\ge 1)$, i.e., $j$ is the $i$-th iteration since the start of phase $\ell$.   Let $d_i:= 2^i  r$, and recall that $s: =\lceil 10/\delta \rceil$. At the end of iteration $j$, the following hold:
\begin{enumerate}
    \item {\bf (Number of finished \LOCAL executions)} 
    \label{item1}
    For every $v\in  H^+_\ell(j-1)$, if there is no $u\in H^+_{\ell-1} (j-1)$ satisfying $\dis_{H^+_{\ell}(j-1)}(u,v)\le sd_i\cdot r$, then we have computed $v$'s state  after $sd_i$ executions of the \LOCAL algorithm (from \Cref{lem:local}) in phase $\ell$.  
    \item {\bf (Radius of collected neighborhoods)} For every $v\in H^+_\ell(j)$, we have collected the following information into one machine:
    \label{item2}
    \begin{alphaenumerate}
        \item The $d_i$-hop neighborhood of $v$ in graph $H^+_\ell(j)$. \label{item:neighbor2}
        \item The communication history during phase $\ell$ along every edge $(u,w)$,  where vertex $u\notin H^+_{\ell}(j)$ was eliminated during phase $\ell$, and $w$ belongs to the neighborhood defined in (\hyperref[item:neighbor2]{a}).
    \end{alphaenumerate}
\end{enumerate}
\end{property}
Informally, \Cref{item1} says that we have simulated $sd_i$ executions of the \LOCAL algorithm (each using $r$ rounds) in phases $\ell$, as long as the $sd_i\cdot r$-hop neighborhood contained no initially pending vertices when iteration $j$ began.

	Now we describe how to implement $\textsc{Initialize}(\ell),\textsc{Simulate}(\ell,j),\textsc{Update},\textsc{Expand}(\ell,j)$ in iteration $j=(\ell-1)t+i$ using $O(1)$ MPC rounds and maintain \Cref{prop:lowarb:invariant} (the purpose of these procedures has been informally described at the end of \cref{sec:lowarb:defini}).
	
	\begin{lemma}
	Assume \cref{prop:lowarb:invariant} holds at the end of iteration $j-1$.  We can implement $\textsc{Simulate}(\ell,j),\textsc{Update},\textsc{Expand}(\ell,j)$ (concurrently for all phases $\ell\in [L_{\sf active}]$, see 4th line in \cref{algo:main}) using $O(1)$ MPC rounds so that at the end of iteration $j$, \Cref{prop:lowarb:invariant} holds for all phases $\ell \in [L_{\sf active}]$.
	\end{lemma}
	\begin{proof}
	 Let $j=(\ell-1)t+1$. Recall that $d_i = 2^i r$. In the following algorithmic description we omit the low-level implementation details of sublinear-memory MPC model. We will briefly address them in Remark~\ref{remark:mpc}.
\begin{itemize}
\item $\textsc{Simulate}(\ell,j)$.\,
We apply \Cref{obs:pending}, where $A$ is the phase $\ell$ algorithm, and vertices in $H_{\ell-1}^+(j-1)$ are initially pending.
To maintain \cref{prop:lowarb:invariant} \cref{item1}, it suffices to simulate algorithm $A'$ (defined in \Cref{obs:pending}) up to $sd_i\cdot r$ LOCAL rounds.

Apply \cref{lemma:died-simulation},  with $U$ being the set of vertices $u\notin H^+_{\ell}(j-1)$ which were eliminated during phase $\ell$.
By \cref{prop:lowarb:invariant} \cref{item2}, we already have the $d_{i-1}$-hop neighborhood of every vertex in $H_{\ell}^+(j-1)$ (as well as the communication history required by \cref{lemma:died-simulation}) stored into one machine. So we can simulate $d_{i-1}$ LOCAL rounds for all $v\in H_{\ell}^+(j-1)$ in one MPC round.

After simulating $d_{i-1}$ LOCAL rounds, we update the current states of the vertices in all stored $d_{i-1}$-hop neighborhoods, using $O(1)$ MPC rounds. Then we again simulate $d_{i-1}$ LOCAL rounds for all vertices in one MPC round, so that the total number of simulated LOCAL rounds becomes $2d_{i-1}$.
And we update the current states of vertices, and so on. 
In this way, we can simulate the first $sd_i\cdot r$ \LOCAL rounds in phase $\ell$  using $O(sd_i\cdot r/d_{i-1}) = O(1)$ MPC rounds. (This is the ``blind coordination'' technique used by Behnezhad et al.~\cite[Section 5.3]{BehneBDFHKU2019Massively})
When we simulate the behaviour of $v$ in phase $\ell$, we also record the communication history of $v$. Since the degree of $v$ in phase $\ell$ is at most $\Delta_{\ell-1}$, the size of the messages is $O(\Delta_{\ell-1} \cdot \polylog n)$.
\item $\textsc{Update}$.\, 
After vertices are removed from $G(j-1)$ by $\textsc{Simulate}(\ell,j)$, let $G(j)$ be the induced subgraph of the remaining vertices, 
and compute the vertex degrees in graph $G(j)$.
 Then, for every $\ell$, delete from $H_{\ell}(j-1)$ the vertices whose degree dropped to $\le \Delta_{\ell}$ and the ones that got removed, and obtain $H_{\ell}(j)$. Then similarly obtain $H_\ell^+(j)$. We also delete these vertices from the stored neighborhoods.
\item $\textsc{Expand}(\ell,j)$.\,
To maintain \cref{prop:lowarb:invariant} \cref{item2}, we perform one graph exponentiation step in $O(1)$ MPC rounds. Every $u\in H^+_\ell(j)$  requests the neighborhood of every other $v$ in the stored $d_{i-1}$-hop neighborhood of $u$.  The new radius becomes $2d_{i-1} = d_i$. Then it is easy to collect the required communication history of vertices in this neighborhood. 
\end{itemize}

Note that $\textsc{Simulate}(\ell,j)$ and $\textsc{Expand}(\ell,j)$ are independent across all phases $\ell$ and can be executed concurrently.
\end{proof}

	\begin{lemma}
	At iteration $j=(\ell-1)t+1$, we can implement $\textsc{Initialize}(\ell)$ using $O(1)$ MPC rounds so that at the end of this iteration \Cref{prop:lowarb:invariant} holds for phase $\ell$.
	\end{lemma}
	\begin{proof}
We  run the \LOCAL algorithm in phase $\ell$ for $sd_1$ times (in the sense of \cref{obs:pending}). 
This can be done in $O(sd_1\cdot r)=O(1)$ MPC rounds with $\tilde O(n+m)$ global memory.
Then, we collect the $2r$-hop neighborhood of every vertex in $H_{\ell}^+(j)$ (and the required communication history) in $O(1)$ MPC rounds.
	\end{proof}
	
	\begin{remark}
\label{remark:mpc}
We did not spell out the low-level details, but it is not hard to verify that our algorithm only involves basic operations and   can be easily implemented using the standard MPC primitives developed in previous works, e.g.\ \cite[Section E]{andonifocs} and the references therein. 
The only concern is that, when $\Delta \ge n^{\delta}$, we cannot gather the neighborhood of a vertex into one machine. 
This issue was already addressed in Behnezhad~et~al.'s unpipelined algorithm  \cite[Section~6]{BehneDHK2018Massively}, using a load-balancing technique.
Hence we can first sequentially run $O(1)$ phases of their algorithm (in $O(\log \log n)$ MPC rounds) at the very beginning, which reduce the maximum degree of the remaining graph to below $n^{\alpha}$ for any desired constant $\alpha>0$.
\end{remark}
	
Now we use \cref{prop:lowarb:invariant} \cref{item1} to show the round complexity of our algorithm.
	\begin{theorem}
	    Our pipelined MPC algorithm runs in $O(\sqrt{\log \alpha} \log \log \alpha + \log\log n)$ rounds.
	\end{theorem}
	\begin{proof}
	    First consider the unpipelined scenario: in phase $\ell$, by \cref{lem:local}, after $\log n$ executions of the \LOCAL algorithm, w.h.p.\ there are no remaining vertices with degree higher than $\Delta_{\ell}$.
	    
	Now we look at our pipelined MPC simulation.
	Note that each iteration takes $O(1)$ MPC rounds.
	After $j_1 = O(t\cdot L) = O(\log \log \Delta)$ iterations, all phases have started their simulation.
	Then, after another $j_2 = O(\log \log n)$ iterations,  every phase $\ell$ has at least been running for $i> \log \log n$ iterations. Let $j=j_1+j_2$ denote the number of iterations so far. By  \cref{prop:lowarb:invariant} \cref{item1}, if $H_{\ell-1}^+(j-1)$ is empty, then at the end of iteration $j$ we will have finished $sd_i>\log n$ many \LOCAL executions of phase $\ell$, which (together with the previous paragraph) implies  $H_{\ell}(j)$ is empty.
	Hence, at this point $H_{1}(j)$ is already empty (since $H_{0}^{+}(j-1)$ is always empty). 
	This in turn implies that in the next iteration $j+1$, $H_{2}(j+1)$ will be empty.
	By a simple induction, after $L-1$ iterations, $H_{L}(j+L-1)$ becomes empty, which means that the maximum degree of the remaining graph $G(j+L-1)$ is at most $\Delta_{L} = \poly(\alpha,\log n)$. 
	Finally, we use Ghaffari and Uitto's algorithm~\cite{GhaffU2019Sparsifying}, which takes $O(\sqrt{\log \Delta_L}\log \log \Delta_L +\sqrt{\log \log n}) \le O(\sqrt{\log \alpha}\log \log \alpha)+\tilde O(\sqrt{\log \log n}) $ MPC rounds.
	Hence, the total round complexity is $O(\sqrt{\log \alpha} \log \log \alpha + \log\log n)$.
	\end{proof}

\subsection{Memory requirement}
\label{sec:lowarb:memory}
In this section we will show that our algorithm satisfies the local memory and global memory constraints.
In the following, we will always assume $j=(\ell-1)t+i$, that is, $j$ is the $i$-th iteration since phase $\ell \in [L]$ starts.
Recall that $\Delta_\ell := \Delta^{1/2^{\ell}}$.
For a subgraph $H$, we use $|H|$ to denote the number of vertices in $H$.
We need the following lemma.
\begin{lemma}
\label{lem:num-vert-H}
With high probability,  $|H^+_{\ell}(j)|\le 2n/\Delta_{\ell-1}^{(s-1)d_i}$.
\end{lemma}
In order to prove \cref{lem:num-vert-H}, we first define the following quantity. 
\begin{definition}
\label{def:D}
For every $v\in H_{\ell}^+ (j)$, let
$\displaystyle D_{\ell,j}(v):=\max_{ \begin{smallmatrix}\scriptstyle u\in G(j-1):\\ \scriptstyle\dis_{G(j-1)}(u,v) \, \le\, (sd_i+2)r\end{smallmatrix}} \big \{\deg_{G(j-1)} (u) \big \}.$
\end{definition}
Note that $D_{\ell,j}(v) \ge \deg_{G(j-1)}(v^-)> \Delta_{\ell}$, where $v^-\in H_{\ell}(j)$ and $\dis_{G(j-1)}(v^-,v)\le r$. Then, we classify the vertices  in $H_{\ell}^+(j)$ by the values of $D_{\ell,j}(v)$.

\begin{definition}
The vertices in $H_{\ell}^+(j)$ are partitioned as $\mathop{\dot{\bigcup}}_{0\le q \le \ell-1}P^{(q)}_{\ell,j}$, where
\[P^{(q)}_{\ell,j} := \big \{ v\in H_{\ell}^+(j): D_{\ell,j}(v) \in (\Delta_{\ell-q}, \Delta_{\ell-q-1}]\big\} .\]
\end{definition}
Informally, for $v\in P_{\ell,j}^{(q)}$, there exists a nearby vertex $u$ which is currently stuck in phase $\ell-q$.
In order to bound $|H_{\ell}^+(j)|$, we will bound the size of $P_{\ell,j}^{(q)}$ separately as follows.
\begin{lemma}
\label{lem:plj}
With high probability, the following upper bounds hold.
\begin{enumerate}
    \item $|P_{\ell,j}^{(0)}|\le n/\Delta_{\ell-1}^{(s-1)d_i}$.
    \label{item101}
    \item For $1\le q\le \ell-1$, $|P_{\ell,j}^{(q)}| \le |H_{\ell-q}(j-1)|\cdot \Delta_{\ell-q-1}^{(sd_i+3)r}$. 
    \label{item102}
\end{enumerate}
\end{lemma}
\begin{proof}[Proof of \Cref{item101}]


For $v^+\in P_{\ell,j}^{(0)}\subseteq H_{\ell}^+(j)$, there exists $v\in H_{\ell}(j)$ such that $\dis_{G(j)}(v,v^+)\le r$.
 Suppose for contradiction that there exists $u^+\in H_{\ell-1}^+(j-1)$ satisfying $\dis_{G(j-1)}(u^+,v) \le sd_i r$. Then there exists $u \in H_{\ell-1}(j-1)$ with $\dis_{G(j-1)}(u,u^+)\le r$, implying that $\dis_{G(j-1)}(u,v^+)\le \dis_{G(j-1)}(u,u^+)+\dis_{G(j-1)}(u^+,v) +\dis_{G(j-1)}(v,v^+)\le r+sd_i r + r = (sd_i+2)r$.
 However, since $\deg_{G(j-1)}(u) > \Delta_{\ell-1}$, this contradicts $D(v^+) \le \Delta_{\ell-1}$.
 Hence such $u^+$ does not exist.
 
 Then, by  \cref{prop:lowarb:invariant} \cref{item1}, this implies that the state of $v$ after $sd_i$ executions of the \LOCAL algorithm in phase $\ell$ is already determined.
 By the property of the \LOCAL algorithm, the number of such vertices $v\in H_{\ell}(j)$ is at most $n/\Delta_{\ell-1}^{sd_i}$.
 
Note that $v^+$ is connected to $v$ in graph $G(j)$ by a path of length at most $r$, and previous discussion implies that every vertex $w$ on this path satisfies $w\notin H_{\ell-1}^+(j-1)$ and thus $\deg_{G(j)}(w) \le  \Delta_{\ell-1}$. 
So the number of such $v^+$ is at most $(n/\Delta_{\ell-1}^{sd_i})\cdot \Delta_{\ell-1}^{r+1} \le n/\Delta_{\ell-1}^{(s-1)d_i}$.
\end{proof}
\begin{proof}[Proof of \Cref{item102}.]
For $v\in P_{\ell,j}^{(q)}$, let $u\in G(j-1)$ be the maximizer in the definition of $D_{\ell,j}(v)$.
Then $v$ is connected to $u$   in graph $G(j-1)$ by a path   of length at most $(sd_i+2)r$, on which every vertex $w$ (including $u,v$) has degree $\deg_{G(j-1)}(w)\le \deg_{G(j-1)}(u) \le \Delta_{\ell-m-1}$.
Fixing a vertex $u$, the number of vertices $v$ that can be reached from $u$ in this way is at most $\Delta_{\ell-q-1}^{(sd_i+2)r+1} \le \Delta_{\ell-q-1}^{(sd_i+3)r}$. 
The number of vertices $u\in G(j-1)$ with $\deg_{G(j-1)}(u)\in (\Delta_{\ell-q},\Delta_{\ell-q-1}]$ is at most $|H_{\ell-q}(j-1)|$.
Hence, we conclude that $|P_{\ell,j}^{(q)}| \le |H_{\ell-q}(j-1)| \cdot \Delta_{\ell-q-1}^{(sd_i+3)r}$.
\end{proof}
\begin{proof}[Proof of \cref{lem:num-vert-H}]
To show that $|H^+_{\ell}(j)|\le 2n/\Delta_{\ell-1}^{(s-1)d_i}$, we prove by induction on $\ell$. 

For $\ell=1$, by \cref{lem:plj}, we have $|H_{1}^+(j)| =|P_{1,j}^{(0)}| \le n/\Delta_{0}^{(s-1)d_i}\le 2n/\Delta_{\ell-1}^{(s-1)d_i}$.

For $\ell \ge 2$, by \cref{lem:plj} and induction hypothesis,
\begin{align*}
    |H_{\ell}^+(j)| =|P_{\ell,j}^{(0)}|+ \sum_{1\le q <\ell}|P_{\ell,j}^{(q)}| 
    &\le \frac{n}{\Delta_{\ell-1}^{(s-1)d_i}} + \sum_{1\le q< \ell}|H^+_{\ell-q}(j-1)|\cdot \Delta_{\ell-q-1}^{(sd_i +3)r}\\
    & \le \frac{n}{\Delta_{\ell-1}^{(s-1)d_i}} + \sum_{1\le q< \ell}\frac{2n}{\Delta_{\ell-q-1}^{(s-1)d_{i+qt-1}}}\cdot \Delta_{\ell-q-1}^{(sd_i+3) r}\\
    & \le \frac{n}{\Delta_{\ell-1}^{(s-1)d_i}} + \sum_{1\le q< \ell}\frac{2n}{\Delta_{\ell-1}^{d_{i}((s-1)2^{qt-1}-(s+2)r)/2^q}}\\
    & \le \frac{n}{\Delta_{\ell-1}^{(s-1)d_i}} + \sum_{1\le q< \ell}\frac{2n}{\Delta_{\ell-1}^{(s-1)d_{i}q+1}},
\end{align*} 
where in the second inequality we used the induction hypothesis and $j-1=(\ell-q-1)t+i+qt-1$,
and in the last inequality we chose $t$ to be a big enough constant. Then, assuming $\Delta_{\ell-1}\ge 4$,  we conclude  
$|H_{\ell}^+(j)| \le 2n/\Delta_{\ell-1}^{(s-1)d_i}$.
\end{proof}

Now we turn to analyzing the memory used for storing the $d_i$-hop neighborhood of $v\in  H_{\ell}^+(j)$ (together with their communication history required by \Cref{prop:lowarb:invariant} \Cref{item2}).
\begin{lemma}
\label{lem:neighborhood}
    For $v\in P_{\ell,j}^{(q)}$ $(0\le q<\ell)$, the memory for storing the $d_i$-hop neighborhood of $v$ in graph $H_{\ell}^+(j)$ is at most 
    $\Delta_{\ell-q-1}^{d_i+2}\cdot \polylog n.$
\end{lemma}
\begin{proof}
Since $v\in P_{\ell,j}^{(q)}$, we have $D_{\ell,j}(v) \le \Delta_{\ell-q-1}$. By the definition of $D_{\ell,j}(v)$, every vertex $u$ in the $d_i$-hop neighborhood of $v$ in $H_{\ell}^+(j)$ has $\deg_{H_{\ell}^+(j)}(u)\le \deg_{G(j-1)}(u)\le  D_{\ell,j}(v) \le \Delta_{\ell-q-1}$,
  and thus the number of vertices $u$ in the $d_i$-hop neighborhood of $v$ is at most $\Delta_{\ell-q-1}^{d_i+1}$. 
 Hence the total memory for storing this neighborhood  (together with their communication history in phase $\ell$) is at most 
$\Delta_{\ell-q-1}^{d_i+1}\cdot (\Delta_{\ell-1}\polylog n) \le \Delta_{\ell-q-1}^{d_i+2}\cdot \polylog n.$
\end{proof}

We prove that this memory does not violate the memory constraints.
\begin{theorem}[Local memory constraint]
 With high probability, for any $v\in H_{\ell}^+(j)$, the memory for storing the $d_i$-hop neighborhood of $v$ in graph $H^+_\ell(j)$ is at most $n^\delta$.
\end{theorem}
\begin{proof}
Let $q$ such that $v\in P_{\ell,j}^{(q)}$.
 Such $v$ exists only if $|P_{\ell,j}^{(q)}|\ge 1$. 
    We consider two cases.
    \begin{enumerate}
        \item If $q=0$, by \cref{lem:plj} and $|P_{\ell,j}^{(0)}|\ge 1$ we have  $n/\Delta_{\ell-1}^{(s-1)d_i}\ge 1$.
        Then by \cref{lem:neighborhood}, the memory is $ 
          \Delta_{\ell-1}^{d_i+2}\cdot \polylog n \le n^{(d_i+2)\big /\left ((s-1)d_i\right )}\cdot \polylog n \ll  n^\delta,
        $ where the last inequality follows from $s\ge 10/\delta$ and $d_i\ge 2$.
    
    \item If $1\le q< \ell$, by  \cref{lem:plj} and $|P_{\ell,j}^{(q)}|\ge 1$ we have $|H_{\ell-q}(j-1)|\ge 1$, which then implies $2n/\Delta_{\ell-q-1}^{(s-1)d_{i+qt-1}} \ge 1$ by \cref{lem:num-vert-H}.  
  Then similarly by \cref{lem:neighborhood}, the memory requirement is
 $  \Delta_{\ell-q-1}^{d_i+2}\cdot \polylog  n  \le (2n)^{(d_i+2)\big /\left ((s-1)d_{i+qt-1}\right )}\cdot \polylog n
       \ll n^\delta. $ \qedhere
    \end{enumerate}
\end{proof}


\begin{theorem}[Global memory constraint]
With high probability, the global memory for storing the  neighborhoods  is $\tilde O(n)$. 
\end{theorem}
\begin{proof}
Since there are only $L=O(\log \log \Delta)$ phases running concurrently, it suffices to show that, for every $\ell \in [L]$, the total memory for storing the $d_i$-hop neighborhoods of all vertices $v$ in graph $H_{\ell}^+(j)$ is at most $\tilde O(n)$, where $j=(\ell-1)t+i$.
 
 By \cref{lem:neighborhood}, the total memory can be bounded by
 \[
 |P_{\ell,j}^{(0)}|\cdot\Delta_{\ell-1}^{d_i+2} \cdot \polylog n+ \sum_{1\le q <\ell}|P_{\ell,j}^{(q)}|\cdot  \Delta_{\ell-q-1}^{d_i+2}\cdot \polylog n\]
\begin{align*}
    \le  & \ \frac{n}{\Delta_{\ell-1}^{(s-1)d_i}}\cdot \Delta_{\ell-1}^{d_i+2}\cdot \polylog n 
    + \sum_{1\le q< \ell}|H^+_{\ell-q}(j-1)|\cdot \Delta_{\ell-q-1}^{(sd_i +3)r}\cdot \Delta_{\ell-q-1}^{d_i+2}\cdot \polylog n\\ 
    \le &  \ \tilde O(n),
\end{align*}
where the last inequality follows from a similar calculation as in the proof of \cref{lem:num-vert-H}.
\end{proof}

\section{4-coloring of trees in $O(\log \log n)$ rounds}
In this section, we present an algorithm that colors a tree with $4$ colors in $O(\log\log n)$ MPC rounds and $O(m)$ global memory, thus providing a proof for \Cref{thm:4-coloring}. We start with a high-level overview, and then fill in the details of the algorithm and its analysis.

\subparagraph*{High-Level Overview} Our $4$-coloring algorithm starts by randomly partitioning the vertices of the input tree into two sets. Each set induces a forest such that each connected component has a diameter of $O(\log n)$, with high probability. Each connected component corresponds to a tree and will be rooted (i.e., orienting each edge towards the root) by using $\Theta(\log^2 n)$ parallel black-box invocations to the connected components algorithm of Behnezhad et al.~\cite{BehnezhadDELM19} (which improved on the earlier work by Andoni et al.~\cite{andonifocs}). We note that this connected components algorithm runs $O(\log D + \log\log n)$ rounds, where $D$ denotes the maximum component diameter, using strongly sublinear local memory and $O(m)$ global memory. Once the tree is rooted, computing a $2$-coloring in $O(\log \log n)$ rounds is easy: we can learn the distance from the root, by pointer jumping along the outgoing edges and leveraging that the diameter is bounded by $O(\log n)$, and then $2$-color nodes based on odd or even distances. Thus, each of the two forests can be colored with $2$ colors. This results in a $4$-coloring of the complete tree. The presented algorithm would need a slightly superlinear global memory of  $\Omega(n\log^2 n)$. To alleviate this issue and work with only $O(n)$ global memory, our actual algorithm first reduces the number of vertices to $O(n/\poly(\log n))$ by using the well-known peeling algorithm~\cite{barenboim2010sublogarithmic}, which in $O(\log\log n)$ iterations repeatedly removes vertices of degree at most $2$. These removed vertices are then colored at the very end of the algorithm, after coloring those $O(n/\poly(\log n))$ remaining vertices, in $\Theta(\log \log n)$ additional MPC rounds. 

\subsection{Details + Analysis}
In this section, we explain and analyze the algorithm in detail. Let $T$ denote the input tree.
\subsubsection{Reducing the number of vertices}
\label{subsec:reducing}
First, we explain and analyze the peeling algorithm which reduces the number of vertices to $O(n/\log^2 n)$. The peeling algorithm consists of $N = \Theta(\log \log n)$ iterations. In each iteration, we remove all vertices that have at most $2$ neighbors in the current graph. In each iteration of the peeling algorithm, a constant fraction of the remaining vertices are removed. 

\begin{lemma}
Let $G = (V,E)$ be a forest. Then, $|\{v \in V \colon \deg(v) \geq 3)\}| \leq \frac{2}{3}|V|$.
\end{lemma}
\begin{proof} We have $3 \cdot |\{v \in V \colon \deg(v) \geq 3)\}| \leq \sum_{v\in V} \deg(v) \leq 2|E| \leq 2|V|,$ where the last inequality follows, as each forest has at most $|V| - 1$ edges and the second-last inequality is the well-known handshaking lemma.
\end{proof}

Thus, after iteration $i$ of the peeling algorithm, the number of remaining vertices is at most $n \cdot (2/3)^i$. Hence, after running $N = \lceil 2 \log_{3/2} \log n \rceil$ iterations of the peeling algorithm, the number of remaining vertices is at most $O(n/\log^2(n))$. Note that each iteration of the peeling algorithm can easily be implemented in $O(1)$ MPC rounds and linear global memory. 

Next, we explain how to color the vertices removed by the peeling algorithm, once all the remaining vertices are colored, using no additional colors. The idea is well-known in the context of LOCAL algorithms. For $i \in [N]$, let $W_i$ denote the set of vertices removed in the $i$-th iteration. The basic idea is to first color all vertices in $W_N$, then the ones in $W_{N-1}$ and so on. By coloring the vertices in that order, the number of neighbors that previously got assigned a color is upper bounded by $2$. Hence, we can assign each node one of the first $3$ colors without creating any conflict. In order for this procedure to be efficient, we need to color multiple nodes simultaneously. However, coloring all nodes in $W_i$ in parallel is problematic, as $W_i$ may contain neighboring nodes. To circumvent this problem, we temporarily color $T[W_i]$ with a constant number of colors---these are not a part of the output coloring but merely used as a schedule in computing the output coloring. In fact, we can color all of $T[W_1], T[W_2], \ldots, T[W_N]$ simultaneously, each separately using constant many colors, in $O(\log ^ * n)$ MPC rounds. This is done by simulating the LOCAL algorithm of Cole and Vishkin~\cite{ColeVishkin}. Then, to compute the output coloring, for each $T[W_i]$, we iterate through the constant many schedule colors in $T[W_i]$, and we compute the output color of all the nodes that got assigned the same schedule color in parallel. Hence, one never assigns two neighboring nodes an output color simultaneously. Coloring the nodes of one schedule color class in $W_i$ can be simulated in $O(1)$ MPC rounds. Thus, we can color all the deleted vertices in $O(1) \cdot O(1) \cdot O(\log \log n)$ MPC rounds and using linear global memory. Hence, what remains is to show how to color a forest consisting of $O(n/\log^2(n))$ nodes in $O(\log \log n)$ MPC rounds, using $O(n)$ global memory.

\subsubsection{Random Partitioning}
\label{subsec:randPartitiong}
In order for our procedure to be efficient, we need each connected component under consideration to have a diameter of $O(\log n)$. We achieve this by randomly partitioning the remaining vertices into two sets $V_1$ and $V_2$, with each remaining vertex being in either one of them with a probability of $1/2$, independent of the other vertices. The partitioning leads to each connected component of $T[V_1]$ and $T[V_2]$ having a diameter of $O(\log n)$, with high probability. The argument is simple: each path of length $\omega(\log n)$ is present in $T[V_1]$ or $T[V_2]$ with a probability less than $1/\poly(n)$. Then, we can union bound over the at most $O(n^2)$ distinct paths in the tree to conclude that no path of length $\omega(\log n)$ will be present in $T[V_1]$ or $T[V_2]$, with probability $1-1/\poly(n)$.

\subsubsection{Handling connected components in parallel}
In the next paragraph, we discuss an algorithm that computes a $2$-coloring of a given connected component in $O(\log \log n)$ rounds. The algorithm works under the assumption that each node of the component has a unique identifier in $[\eta]$, with $\eta$ denoting the size of the connected component. Furthermore, it assumes that the edges of the connected component are the only input given to the algorithm. It requires each machine to have a local memory of $O(\min(\eta, n^{\delta}))$ and uses $O(\eta \log^2n)$ global memory. Before going into the details of how the algorithm works, we first argue why the existence of such an algorithm readily implies that we can $2$-color each connected component in $T[V_1]$ and $T[V_2]$ with $2$ colors in $O(\log \log n)$ MPC rounds and $O(m)$ global memory. 

Before $2$-coloring each connected component, we start by computing the connected components of both $T[V_1]$ and $T[V_2]$ by using the connected component algorithm of \cite{BehnezhadDELM19}. The algorithm runs in $O(\log \log n  + \log D)$ rounds, where $D$ denotes the maximal component diameter. We can assume $D = O(\log n)$. Thus, we can find the connected components of both $T[V_1]$ and $T[V_2]$ in $O(\log \log n)$ rounds. 
After having computed the connected components,  we create one tuple per node with the first entry being equal to its component identifier and the second entry being equal to the identifier of the node itself. 
By sorting these tuples lexicographically in $O(1)$ rounds using the algorithm of \cite{goodrich2011sorting}, we can identify the size of each connected component. Furthermore, one can assign each connected component a number of machines proportional to its size, such that the total memory capacity of the assigned machines is $O(n)$.  
That is, each large connected component with $\eta \geq n^{\delta}$ nodes gets assigned $\Theta(\log^2 (n) \eta / n^\delta)$ many machines. 
After relabeling the vertices of the connected component with unique identifiers between $1$ and $\eta$, the machines assigned to such a large component can $2$-color the component in $O(\log \log n)$ rounds.  
Each small component with less than $n^\delta$ vertices is stored on one machine, which may be responsible for multiple small components.
A $2$-coloring of such a small component can be computed locally on the corresponding machine. 
 
\subsubsection{Rooting and $2$-coloring a tree with diameter $O(\log n)$}
Next, we focus on a single connected component in either $T[V_1]$ or $T[V_2]$ and show how to root the connected component by using $O(\log^2 n)$ invocations of the connected component algorithm in parallel. As stated in the previous paragraph, we assume that each node of the connected component has a unique identifier in $ [\eta]$, with $\eta$ denoting the size of the component. We pick an arbitrary node as the root, i.e., the root with identifier $1$. We want each node, except for the root, to learn which neighboring node is its parent. Our algorithm relies on the following simple, but crucial observation: the parent of a node is the only neighboring node that, when deleting a subset of the edges in the tree, can still remain in the same connected component as the root, while the node itself got disconnected from the root. 

To make use of this observation, we remove each edge independently with a probability of $1 / \log(n)$. Let $v$ be an arbitrary node. Notice that if $v$ is disconnected from the root but a neighbor $u$ of $v$ is connected to the root, then one can deduce that $u$ is the parent of $v$.  Let $p(v)$ be the parent of $v$. Node $p(v)$ remains in the same connected component as the root if all of the at most $O(\log n)$ edges on the  path from $p(v)$ to the root remain in the graph. This happens with a probability of at least $(1-1/\log n)^{O(\log n)} = \Omega(1)$. Conditioning on this event, node $v$ gets disconnected from the root if the edge between $v$ and $p(v)$ gets removed, which happens with a probability of $1/\log(n)$. If both of these happen, we have a good event and $v$ can identify its parent $p(v)$. Thus, we can conclude that a fixed vertex $v$ can determine its parent with a probability of $\Omega(1/\log(n))$. 

By running this procedure $\Theta(\log^2 n)$ times in parallel (and independently), we can conclude---by a Chernoff bound and a union bound over all vertices---that each vertex determines its parent with high probability. Hence, we can $2$-color the connected component using $O(\min(\eta, n^\delta))$ memory per machine and $O(\log^2(n) \eta)$ global memory. 

Once we have this orientation towards the root, and since the diameter of the tree is $O(\log n)$, each node can compute the distance of it from the root using $\Theta(\log \log n)$ steps of pointer jumping, as we explain next.

\begin{lemma}
Given a tree of depth $O(\log n)$ with $\eta$ nodes, which is oriented towards the root, we can compute the distance of each node from the root in $O(\log \log n)$ rounds, using $O(\min(\eta, n^{\delta}))$ local memory and $O(\eta \log^2 n)$ global memory.    
\end{lemma}
\begin{proof}
We use a \emph{pointer jumping} idea, along with some sorting subroutines of MPC. For each node $v$, we keep track of one pointer $p(v)$ which we initially set equal to its parent node (for the root $r$ the pointer points to the node itself). Also, we define the ancestor list $AL(v)$, which  is a set that initially only contains the node $v$.

Then, each iteration of pointer jumping consists of two steps: first, each node $v$ updates its ancestor list as $AL(v)\gets AL(v) \cup AL(p(v))$. Then, in the second step, we redefine $p(v)\gets p(p(v))$. It is clear that after $O(\log\log n)$ iterations, $p(v)$ is equal to the root of its component and $AL(v)$ is equal to the set of all of its ancestors, out of which $v$ learns its distance to the root.

Implementing the pointer jumping needs some care. In each iteration of pointer jumping, each node $v$ needs to learn the value of $p(u)$, where $u=p(v)$. While node $v$ wants to learn only one value, node $u$ might have to inform many nodes about the value of $p(u)$. Essentially the same procedure needs to be done for learning $AL(u)$, so we focus on $p(u)$.

We use a basic sorting subroutine \cite{goodrich2011sorting}, which runs in constant time. For the purpose of sorting, for each node $v$, define an item $\langle p(v), v\rangle$. Here, $v$ and $p(v)$ are the respective identifiers, which are numbers in $\{1, \dots, \eta\}$. Moreover, add for each node $v$ two extra items $\langle v, -\inf\rangle $ and $\langle v, +\inf\rangle $. Sort all these items lexicographically. At the end of the sorting, each machine that holds an item knows the rank of this item in the sorted list.

Once we have the items sorted, for each node $u$, in the sorted order, the items start with $\langle u, -\inf\rangle $, end with $\langle u, +\inf\rangle $, and in between these two are the entries of all nodes $v$ such that $p(v)=u$. As a result, the machine that holds node $u$ and generated the two items $\langle u, -\inf\rangle $ and $\langle u, +\inf\rangle $ knows the ranks of these two items in the sorted order. 

Now, split the items among the machines, in an ordered way, so that the $i^{th}$ machine receives the items with rank $[S(i-1)+1, S i]$. Finally, use a broadcast tree of constant depth~\cite{goodrich2011sorting} so that the machine that holds node $u$ informs the machines that, in the sorted order of items, hold items with $u$ as the first entry, about the value of $p(u)$. This way, any machine that holds an item $\langle p(v), v\rangle $ where $p(v)=u$ learns $p(u)$. Hence, that machine can add this value $p(u)$ as a third field to create $\langle p(v), v, p(u)\rangle $. At the very end, we send these three-entry messages back to the machine that initially held $v$. Hence, node $v$ now knows the value of $p(u)$ where $u=p(v)$. 

The same process can be used so that node $v$ learns $AL(p(u))$, where $u=p(v)$. In that case, the message has $O(\log n)$ words --- the maximum number of ancestors in an $O(\log n)$ depth tree. The algorithm uses $O(\min(\eta, n^\delta))$ local memory and $O(\eta \log^2 n )$ global memory, thus proving the lemma.
\end{proof}

Once the nodes know their distance from their respective root, a $2$-coloring is immediate: nodes at odd distances get one color and nodes at even distances get the other. This gives a separate $2$-coloring of each of $T[V_1]$ and $T[V_2]$. Overall, this leads to a $4$-coloring of $T[V_1 \cup V_2]$.  

\subparagraph*{Remark about a work of Brandt et al.~\cite{brandt2019breaking}}

The main result of Brandt et al.~\cite{brandt2019breaking} is a $\Theta(\log^3\log n)$ algorithm for finding a Maximal Independent Set on trees. As a subroutine, they used a simple, deterministic (graph exponentiation style) algorithm to root a tree with a diameter of $D$ in $O(\log D \log \log n)$ rounds, under the assumption of initially having $\Omega(D^3)$ global memory per node. Using their rooting algorithm as a black-box in our approach, we can directly improve this complexity. Namely, after the step of reducing the number of vertices to $O(n/\polylog n)$, as explained in \Cref{subsec:reducing}, and randomly splitting the vertices into two sets $V_1$ and $V_2$, as explained in \Cref{subsec:randPartitiong}, instead of our proposed rooting algorithm, one could invoke the rooting procedure of Brandt et al.~\cite{brandt2019breaking}, which runs in $\Theta(\log^2\log n)$ rounds. Overall, this gives a $4$-coloring in $O(\log^2\log n)$ rounds. Having this coloring, we can easily solve MIS in $O(1)$ additional rounds, simply by iterating through color classes of nodes and greedily adding nodes to the MIS.

\bibliographystyle{plainurl}

\bibliography{ref}
\end{document}